\documentclass[12pt,oneside]{amsart}
\usepackage{amsfonts}
\usepackage{yfonts}
\usepackage{amsthm}
\usepackage{amssymb}
\usepackage{amsmath}
\theoremstyle{plain}
    \newtheorem{theorem}{Theorem}[section]
\newtheorem{proposition}[theorem]{Proposition}
    \newtheorem{lemma}[theorem]{Lemma}
    
   \theoremstyle{definition}     
 \newtheorem{definition}[theorem]{Definition}
      \theoremstyle{remark}
 \newtheorem{remark}[theorem]{Remark}

\newcounter{smallarabics}
\newenvironment{arabicenumerate}
{\begin{list}{{\normalfont\textrm{(\arabic{smallarabics})}}}
  {\usecounter{smallarabics}\setlength{\itemindent}{0cm}
   \setlength{\leftmargin}{5ex}\setlength{\labelwidth}{4ex}
   \setlength{\topsep}{0.75\parsep}\setlength{\partopsep}{0ex}
   \setlength{\itemsep}{0ex}}}
{\end{list}}

\newcounter{smallroman}

\newcommand{\ben}{\begin{arabicenumerate}}  
\newcommand{\een}{\end{arabicenumerate}}

\newcommand{\bex}{\begin{example}}
\newcommand{\eex}{\end{example}}
\def\bel{\begin{lemma}}
\def\eel{\end{lemma}}
\def\bet{\begin{theoreme}}
\def\eet{\end{theoreme}}
\def\bed{\begin{definition}}
\def\eed{\end{definition}}
\def\ber{\begin{remark}}
\def\eer{\end{remark}}

\def\cc{{\mathbb C}}

\def\part{{\rm par}}

\def\Im{{\rm Im}}
\def\Re{{\rm Re}}

\def\bar{\overline}

\def\c0inf{C_0^\infty}

\def\s{{\rm s}}
\def\sa{{\rm s/a}}

\def\fG{{\frak G}}

\def\a{{\rm a}}

\def\i{{\rm i}}

\def\Dom{{\rm Dom}}

\def\loplus{\mathop{\oplus}\limits}

\def\12{\frac{1}{2}}
\def\14{\frac{1}{4}}

\def\e{{\rm e}}

\def\loplus{\mathop{\oplus}\limits}

\def\12{\frac{1}{2}}

\def\e{{\rm e}}

\def\bep{\begin{proposition}}
\def\eep{\end{proposition}}

\def\s{{\rm s}}

\def\beq{\begin{equation}}
\def\eeq{\end{equation}}

\def\CARal{{\rm C\hskip 0.25 em \hbox{\raise 1.72 ex 
\hbox{$\scriptscriptstyle\rm al$}\kern -0.57 em A}R}}

\def\otimesal{\mathop{\hbox{\raise 1.5 ex
  \hbox{$\scriptscriptstyle\rm al$}
\kern -0.92 em \hbox{$\otimes$}}}}
\def\oplusal{\mathop{\hbox{\raise 1.5 ex
  \hbox{$\scriptscriptstyle\rm al$}
\kern -0.92 em \hbox{$\oplus$}}}}
\def\Gammal{\hbox{\raise 1.68 ex 
\hbox{$\scriptscriptstyle\rm al$}\kern -0.50 em $\Gamma$}}
\def\Bal{\hbox{\raise 1.68 ex 
\hbox{$\scriptscriptstyle\rm  al$}\kern -0.50 em $B$}}
\def\CARal{{\rm C\hskip 0.25 em \hbox{\raise 1.72 ex 
\hbox{$\scriptscriptstyle\rm al$}\kern -0.57 em A}R}}

 \makeatletter
      \def\@setcopyright{}
      \def\serieslogo@{}
      \makeatother
\begin{document}

\author{Jan Derezi\'{n}ski}
   \address[J. Derezi\'{n}ski]{Department of Mathematical Methods in Physics,
     Faculty of Physics, University of
     Warsaw, Ho{\.z}a 74, 00-682 Warszawa, Poland}
   \email{Jan.Derezinski@fuw.edu.pl}
\author{Marcin Napi\'{o}rkowski}
   \address[M. Napi\'{o}rkowski]{Department of Mathematical Methods in Physics,
     Faculty of Physics, University of
     Warsaw, Ho{\.z}a 74, 00-682 Warszawa, Poland}
   \email{Marcin.Napiorkowski@fuw.edu.pl}
\author{Jan Philip Solovej}
   \address[J. P. Solovej]
{Department of Mathematics, University of Copenhagen, Universitetsparken 5,
  2100 Copenhagen, Denmark}
\email{solovej@math.ku.dk}
\title[]{On the minimization of  Hamiltonians\\ over pure Gaussian states}

\begin{abstract}
A  Hamiltonian defined as a polynomial in creation
and annihilation operators is considered. After a minimization of
 its expectation value 
over pure Gaussian states,  the Hamiltonian is Wick-ordered in
creation and annihillation operators adapted to the
minimizing  state. It is shown that this procedure 
 eliminates from the Hamiltonian terms of degree 1 and 2 that do not
preserve the particle number, and leaves only the terms that can be interpreted
as quasiparticles excitations.  We propose to call this fact {\em Beliaev's
    Theorem}, since to our knowledge it was mentioned for the first time in a
  paper by Beliaev from 1959. 
   \end{abstract}
\maketitle
 \section{Introduction} Various phenomena in  many-body quantum
 physics are  explained with help of  {\em quasiparticles}. 
Unfortunately, we  are not aware of a rigorous definition of
this concept, except for some very special cases. 

 A typical situation when
one speaks about quasiparticles seems to be the following: Suppose that the
Hamiltonian of a system can be written as
$
H=H_0+V$, where
  $H_0$ is in some sense  dominant
 and $V$ is a perturbation
that in first approximation can be neglected. 
Suppose also that 
\beq
H_0=B+\sum_i \omega_ib_i^*b_i,\label{form}\eeq
  where $B$ is a number, operators $b_i^*$/$b_i$ 
satisfy the standard canonical commutation/anticommutation relations
(CCR/CAR) and the Hilbert space contains a state annihilated by $b_i$ (the
{\em Fock vacuum for $b_i$}).
We then say that the  operators 
 $b_i^*$/$b_i$ {\em create/annihilate a quasiparticle}.

Of course, the above definition is very vague. 

 In our
paper we describe a simple theorem that for many Hamiltonians  gives a natural
decomposition $H=H_0+V$ with $H_0$ of the form 
(\ref{form}), and thus suggests  a possible definition of a
quasiparticle. 
Our starting point is a fairly general Hamiltonian $H$ defined on a bosonic or
fermionic Fock space. For simplicity we assume that the $1$-particle space is finite
dimensional. With some technical assumptions, the whole picture should be
easy to generalize to the infinite dimensional case. 
We
assume that the Hamiltonian is a polynomial in creation and annihilation
operators $a_i^*$/$a_i$, $i=1,\dots,n$. 
(This is a typical assumption in Many Body Quantum Physics and
Quantum Field Theory). 

An important role in Many Body Quantum Physics is played by the so-called {\em
 Gaussian states}, called also {\em quasi-free states}. Gaussian states can be
{\em pure} or {\em mixed}. 
The former are typical for the zero temperature,
whereas the latter for positive temperatures. In our paper we do not consider
mixed Gaussian states.
 
Pure Gaussian states are
 obtained by applying  Bogoliubov transformations to the Fock vacuum state
 (given by the vector 
$\Omega$ annihilated by $a_i$'s). Pure Gaussian states are especially convenient for
computations. 

We minimize the expectation value of the Hamiltonian $H$ with 
respect to pure Gaussian states, obtaining a state given by a vector
$\tilde\Omega$. By 
applying an appropriate Bogoliubov transformation, we can replace the old
creation and annihilation operators $a_i^*$, $a_i$ by new ones 
 $b_i^*$, $b_i$, which are adapted to the ``new vacuum'' $\tilde\Omega$, i.e.,
 that
satisfy $b_i\tilde\Omega=0$. We can rewrite the Hamiltonian $H$ in the new
operators and Wick order them, that is, put $b_i^*$ on the left and $b_i$ on
the right. The theorem that we prove says that
\[H=B+\sum_{ij} D_{ij}b_i^*b_j+V,\]
where $V$ has only terms of the order greater than $2$. In  particular, $H$
does not contain terms of the type $b_i^*$, $b_i$,  $b_i^*b_j^*$, or
$b_ib_j$. It is thus natural to set $H_0:=B+\sum_{ij}
D_{ij}b_i^*b_j$. $D_{ij}$ is a hermitian matrix. Clearly, it can be
diagonalized, so that $H_0$ acquires the form of (\ref{form}).

We present several
 versions of this theorem. First we assume that the Hamiltonian
is even. In this case
 it is natural to restrict the minimization to even pure Gaussian
states. In the fermionic case, we can also minimize over odd pure Gaussian states. 
In the bosonic case, we consider also Hamiltonians without the
evenness assumption, and then we minimize with respect to all pure 
Gaussian states.

 The procedure of minimizing over Gaussian states is widely
applied in practical computations and is known under many names. In the
fermionic case in the contex of nuclear physics
it often goes under the name of the {\em Hartree-Fock-Bogoliubov
method} \cite{RS}.  It 
 is closely related to the {\em Bardeen-Cooper-Schrieffer
approximation} used in superconductivity
\cite{1} and the {\em Fermi liquid theory}
developed by Landau \cite{6}. 
In the bosonic case it is closely related to the {\em Bogoliubov approximation} used
in the theory of superfluidity
\cite{3}, see also \cite{9,4}. 
In both bosonic and fermionic cases it is often called the
{\em mean-field approach} \cite{FW}.

 The fact  that we describe in our paper is probably very well known, at
least on the intuitive level, to many physicists, especially in condensed
matter theory.
One can probably say that it
 summarizes in abstract terms one of
the most widely used methods of contemporary quantum physics.
The earliest reference that we know to a statement similar 
to our main result is formulated in a paper of Beliaev \cite{Be}. Beliaev
studied fairly general femionic Hamiltonians by what we would nowadays call 
the Hartree-Fock-Bogoliubov approximation. In a footnote on page 10 he writes:

\noindent
{\em The condition $H_{20}=0$ may be easily shown to be exactly equivalent to
  the requirement of a minimum ``vacuum'' energy $U$. Therefore, the ground
  state of the system in terms of new particles is a ``vacuum'' state. The
  excited states are characterized by definite numbers of new particles,
  elementary excitations.} 

\noindent Therefore, we propose to call the main result of our paper {\em
  Beliaev's Theorem}.

The proof of Beliaev's Theorem is not
 difficult, especially when it is is formulated in an abstract way,
as we do. Nevertheless, in concrete situations, when similar computations are
 performed, consequences 
of this result   may often 
appear somewhat miraculous. The authors of this work
witnessed it several times: the authors themselves, or their colleagues, after
tedious computations and numerous mistakes watched the unwanted terms
disappear \cite{4,DMN}. As we show, these terms have to disappear by a general
argument.

\medskip
{\small {\bf Acknowledgement.} J.~D. thanks V.~Zagrebnov for useful
 discussions.

 The work of J.~D.
was supported in part by the grant N N201 270135 of
the Polish Ministry of Science and Higher Education. 
J.~D. and J.~P.~S. thank the Danish Council for Independent Research for support
during a visit in the Fall of 2010 of J.~D. to the Department of Mathematics,
University of Copenhagen. 
The work of M.~N.  was supported by the Foundation for Polish Science
International PhD Projects Programme co-financed by the EU  within the
Regional Development Fund}.

\section{Preliminaries}
\subsection{2nd quantization}

We will consider in parallel the bosonic and fermionic case.

Let us describe our notation concerning the 2nd quantization. 
We will always assume that the 1-particle space is $\cc^n$. (It is easy
to extend our analysis to the infinite dimensional case).  The {\em bosonic
Fock space} will be denoted $\Gamma_\s(\cc^n)$ and the {\em fermionic Fock space}
$\Gamma_\a(\cc^n)$.
We use the notation $\Gamma_\sa(\cc^n)$ for either the bosonic or fermionic
Fock space. $\Omega\in \Gamma_\sa(\cc^n)$ stands for the {\em Fock vacuum}.
If $r$ is an operator on $\cc^n$, then $\Gamma(r)$ stands for its {\em 2nd
quantization}, that is
\[\Gamma(r):=\left(\loplus_{n=0}^\infty r^{\otimes n}\right)\Big|_{\Gamma_\sa(\cc^n)}.\]
$a_i^*$, $a_i$ denote the standard {\em creation} and {\em annihilation operators} on
$\Gamma_\sa(\cc^n)$, satisfying the usual canonical
commutation/anticommutation relations.

\subsection{Wick quantization} 

Consider an arbitrary  polynomial on $\cc^n$, that is a function of the form
\begin{align}
 h(\bar{z},z):=\sum_{\alpha,\beta} h_{\alpha,\beta}\bar{z}^{\alpha}z^{\beta}, \label{wielomian}
\end{align}
where $z=(z_{1},\ldots, z_{n})\in \mathbb{C}^{n}$, $\bar z$ denotes the
complex conjugate of $z$  and 
 $\alpha=(\alpha_{1},\ldots,\alpha_{n})
 \in (\mathbb{N}\cup\{0\})^{n}$  represent multiindices.
In the bosonic/fermionic case we  always assume that the coefficients
$h_{\alpha,\beta}$ are  symmetric/antisymmetric separately in the
indices of $\bar z$ and $z$.

We write
 $|\alpha|=\alpha_{1}+\cdots+\alpha_{n}.$ 
We say that $h$ is {\em even} if the sum  in  (\ref{wielomian})
is restricted to even
$|\alpha|+|\beta|$ .

{\em The Wick quantization} of (\ref{wielomian})
is  the operator on $\Gamma_\sa(\cc^n)$ defined as
\begin{align}
h(a^{*},a):=\sum_{\alpha, \beta} h_{\alpha,
  \beta}(a^{*})^{\alpha}a^{\beta}. \label{pqi} 
\end{align}
In the fermionic case, (\ref{pqi}) defines a bounded operator on
 $\Gamma_\a(\cc^n)$. In the bosonic case,  (\ref{pqi}) can be viewed as an
operator on $\bigcap\limits_{n>0}\Dom N^n\subset \Gamma_\s(\cc^n)$, where 
\[N=\sum_{i=1}^na_i^*a_i\]
is the number operator.

\subsection{Bogoliubov transformations} 

We will now present some basic well known facts
about Bogoliubov transformations. For proofs and additional information we
refer to \cite{2} (see also \cite{5}, \cite{F}). We will often use the 
summation convention of summing with respect to repeated indices.

Operators of the form 
\begin{align}
Q=\theta_{ij}a_{i}^{*}a_{j}^{*}+h_{kl}a_{k}^{*}a_{l}+\bar{\theta}_{ij}
a_{j}a_{i}, 
\label{quadratic}
\end{align}
where $h$ is a self-adjoint matrix, will be called \emph{quadratic Hamiltonians}.
In
the bosonic/fermionic case we can always assume  that 
$\theta$ is symmetric/antisymmetric. The group generated by operators of the form $\e^{\i Q}$, where $Q$ is a
quadratic Hamiltonian, is called the \emph{metaplectic (Mp)} group in the
bosonic case and the \emph{Spin} group in the fermionic case.

In the bosonic case, the group generated by  $Mp$ together with
$\e^{\i( y_{i}a_i^*+\bar y_{i} a_i)}$, $y_{i}\in\cc$,
$i=1,\dots,n$, is called the
\emph{affine mataplectic (AMp)}
group.

In the fermionic case, the goup generated by operators 
$y_i a_i^*+\bar{y}_{i}a_i$ with $\sum |y_i|^2=1$ (which are unitary) is called
the \emph{Pin} group. Note that $Spin$ is a subgroup of $Pin$ of index $2$.

In the bosonic case, consider $U\in\emph{AMp}$. It is well known that
\begin{align}
Ua_{i}U^{*}=p_{ij}a_{j}+q_{ij}a_{j}^{*}+\xi_i, \quad Ua^{*}_{i}U^{*}=\bar{p}_{ij}a^{*}_{j}+\bar{q}_{ij}a_{j}+\bar\xi_i  \label{linear}
\end{align} 
for some matrices $p$ and $q$ and a vector $\xi$. 

In the fermionic case, consider $U\in\emph{Pin}$. Then
\begin{align}
Ua_{i}U^{*}=p_{ij}a_{j}+q_{ij}a_{j}^{*}, \quad Ua^{*}_{i}U^{*}=\bar{p}_{ij}a^{*}_{j}+\bar{q}_{ij}a_{j}  \label{linear1}
\end{align} 
for some matrices $p$ and $q$.

The maps
(\ref{linear}) and (\ref{linear1}) 
 are often called {\em Bogoliubov transformations}.
Bogoliubov transformations can be interpreted as automorphism of the
corresponding {\em classical phase space}. Let us describe briefly
this interpretation.

Consider the space $\cc^n\oplus\cc^n$. It
 has a distinguished $2n$-dimensional real subspace
consisting of vectors $(z,\bar z)=\left((z_i)_{i=1,\dots,n},(\bar
z_i)_{i=1,\dots,n}\right)$, which we will call {\em the real part of
$\cc^n\oplus\cc^n$}, and which can be interpreted as the classical
phase space.
The real part of $\cc^n\oplus\cc^n$ is equipped with a
symplectic form 
\begin{equation}
(z,\bar z)\omega(z',\bar z'):=\Im (  z|z'),\label{form1}\end{equation}
and a scalar product
\begin{equation}
(z,\bar z)\cdot(z',\bar z'):=\Re ( z|z').\label{form2}\end{equation}

Consider the bosonic case. 
Note that   the transformation (\ref{linear}),
viewed as a map on the real part of $\cc^n\oplus\cc^n$ given by 
the matrix
$\left[\begin{array}{cc}
p&q\\\bar q&\bar p
\end{array}\right]$ and the vector  $\left[\begin{array}{c}
\xi\\\bar \xi
\end{array}\right]$, preserves the symplectic form (\ref{form1})
-- in other words, it belongs to $ASp$, the
{\em affine symplectic group}. More precisely, it is easily checked
that in this way
we obtain a 2-fold covering homomorphism of $AMp$ onto $ASp$.

In the fermionic case there is an analogous situation.
The transformation  (\ref{linear1}),
viewed as a map on the real part  of $\cc^n\oplus\cc^n$ given by the matrix
$\left[\begin{array}{cc}
p&q\\\bar q&\bar p
\end{array}\right]$, preserves the scalar product (\ref{form2})
-- in other words, it belongs to $O$, the
{\em orthogonal group}. More precisely, it is easily checked
that in this way
we obtain a 2-fold covering homomorphism of $Pin$ onto $O$.

\subsection{Pure Gaussian states}
We will use  the term {\em pure state} to denote a normalized vector modulo
a phase factor.  In particular, we
will distinguish between a pure state and its {\em vector  representative}.


On Fock spaces we have a distinguished pure 
state called the  {\em (Fock) vacuum state},
corresponding to 
$\Omega$. 
States given by vectors
 of the form $U\Omega,$ where $U\in Mp$ or $U\in Spin$, 
will be called \emph{even pure Gaussian states}.
 The family of even pure Gaussian states will be denoted by $\fG_{\sa,0}$.

In the bosonic case, states given by vectors
 of the form $U\Omega$ where $U\in AMp$ will be
called \emph{Gaussian pure states}. 
 The family of bosonic pure Gaussian states will be denoted by $\fG_{\s}$.

In the fermionic case, states given by vectors of the form $U\Omega$, where
$U\in Pin$ will be called \emph{fermionic pure Gaussian states}. The family of
fermionic pure Gaussian states is denoted $\fG_\a$.

Fermionic pure 
Gaussian states that are not even will be called {\em odd  fermionic
pure  Gaussian states}. The family of
odd fermionic pure Gaussian states is denoted $\fG_{\a,1}$.

One can ask whether pure  Gaussian states have {\em  natural}
 vector representatives
(that is, whether one can naturally fix the phase factor of their vector
 representatives).  In the
bosonic case this is indeed always possible.
If $c=[c_{ij}]$ is a symmetric matrix satisfying $\|c\|<1$, then the vector
\beq\det(1-c^*c)^{1/4}\e^{\frac12c_{ij} a_i^*a_j^*}\Omega\label{bosgau}\eeq
defines a state in $\fG_{\s,0}$ (see \cite{10}). 
If 
$\theta=[\theta_{ij}]$ is a symmetric matrix satisfying
$c=\i\frac{\tanh\sqrt{ \theta\theta^*}}{\sqrt{\theta\theta^*}}\theta$, then
 (\ref{bosgau}) equals
\beq \e^{\i X_\theta}\Omega \label{bosgau1} \eeq 
with
\beq X_\theta:=\theta_{ij}a_i^*a_j^*+\bar\theta_{ij} a_ja_i.\eeq

 Each state in  $\fG_{\s,0}$ is represented
uniquely as (\ref{bosgau}) (or equivalently as (\ref{bosgau1})).
In particular,
 (\ref{bosgau1}) provides
 a smooth parametrization of $\fG_{\s,0}$ by symmetric matrices.

The manifold of fermionic even pure Gaussian states is more complicated.
 We will say that a  fermionic even pure 
Gaussian state given by $\Psi$ is {\em
    nondegenerate} if $(\Omega|\Psi)\neq0$ (if it has a nonzero overlap with
  the vacuum). Every nondegenerate fermionic even pure Gaussian state can be
  represented by a vector
\beq \det(1+c^*c)^{-1/4}\e^{\frac12c_{ij} a_i^*a_j^*}\Omega,\label{fergau}\eeq
where  $c=[c_{ij}]$ is an antisymmetric matrix.
If 
$\theta=[\theta_{ij}]$ is an antisymmetric matrix satisfying
$c=\i\frac{\tan\sqrt {\theta\theta^*}}{\sqrt{\theta\theta^*}}\theta$,
$\|\theta\|<\pi/2$, then
 (\ref{fergau})  equals
\beq \e^{\i X_\theta}\Omega \label{fergau1}\eeq 
with
\beq X_\theta:=\theta_{ij}a_i^*a_j^*+\bar\theta_{ij} a_ja_i.\eeq

Not all even fermionic pure Gaussian states are nondegenerate. 
{\em  Slater determinants} with an even nonzero 
number of particles are examples of even
Gaussian pure states that are not 
nondegenerate.
Note  that 
 vectors (\ref{fergau}) are natural representatives of their states. It is
 easy to see  that only nondegenerate
 fermionic pure Gaussian states possess natural
 vector representatives.

Nondegenerate pure
Gaussian states form  an open dense subset of $\fG_{\a,0}$
 containing the Fock state
(corresponding to $c=\theta=0$).
In particular, (\ref{fergau}) provides
 a smooth parametrization of a neighborhood of the Fock state in 
$\fG_{\a,0}$ by antisymmetric matrices.

The fact that each even bosonic/nondegenerate fermionic pure Gaussian state
can be represented by a vector of the form
(\ref{bosgau})/(\ref{fergau}) goes under the name of the {\em Thouless
  Theorem}. (See \cite{Tho}; this name is used eg. in the monograph by Ring and
Schuck \cite{RS}). The closely
related fact saying that these vectors can be represented in the form 
(\ref{bosgau1})/(\ref{fergau1}) is sometimes called the {\em Ring-Schuck
  Theorem.}

By definition, the group AMp/Pin acts transitively on $\fG_{\sa}$.
In other words, for any $\tilde\Omega
\in \fG_{\sa} $ we can find $U\in AMp/Pin$ such
that $\tilde\Omega=U\Omega$. Such a $U$ is not defined uniquely -- it can be replaced
by $U\Gamma(r)$, where $r$ is unitary on $\cc^n$.

Clearly, if we set 
\begin{align}
b_{i}:=Ua_{i}U^{*}, \quad b^{*}_{i}:=Ua^{*}_{i}U^{*}, \label{proper}
\end{align} 
then $b_i\tilde\Omega=0$,
 $i=1,\dots,n$, and they satisfy the same CCR/CAR as $a_i$,
$i=1,\dots,n$. 
If $h$ is a polynomial of the form (\ref{wielomian}), then we can Wick
quantize it using the transformed operators:
\begin{align}
h(b^{*},b)=\sum_{\alpha, \beta} h_{\alpha, \beta}(b^{*})^{\alpha}b^{\beta}. \nonumber
\end{align}
Obviously, $Uh(a^*,a)U^*=h(b^*,b)$.

 \section{Main result}
  As explained in
the introduction, we think that the following result should be called {\em
  Beliaev's Theorem}. 

\begin{theorem} 
Let $h$ be a polynomial on $\cc^n$ and $H:=h(a^*,a)$ its Wick
quantization. We consider the following functions:
\ben 
\item (bosonic case, even pure Gaussian states)
$
\fG_{\s,0}\ni \Phi\mapsto (\Phi|H\Phi)$;
\item (bosonic case, arbitrary pure Gaussian states)
$\fG_{\s}\ni\Phi\mapsto (\Phi|H\Phi)$;
\item (fermionic case, even pure Gaussian states)
$\fG_{\a,0}\ni\Phi\mapsto (\Phi|H\Phi)$;
\item (fermionic case, odd pure Gaussian states)
$\fG_{\a,1}\ni \Phi\mapsto (\Phi|H\Phi)$.
\een
In (1), (3) and (4) we assume in addition that the polynomial $h$
is even. 

For a vector $\tilde\Omega$ representing a pure Gaussian state, 
let $U\in AMp/Pin$ satisfy $\tilde\Omega=U\Omega$. Set $b_i=Ua_i U^*$ and suppose that $ \tilde h$ is the polynomial satisfying
$H=\tilde h(b^*,b)$.
Then the following statements are equivalent:
\medskip

\noindent (A)
 $\tilde\Omega$ represents a stationary point of the function defined in
(1)--(4). 

\medskip

\noindent (B)
\begin{align}
\tilde{h}(b^{*},b)=B+D_{ij}b^{*}_{i}b_{j}+\text{\emph{terms of higher order
    in $b$'s}}. \nonumber 
\end{align}
\label{main}\end{theorem}

\begin{proof}
Let us prove the case (2), which is a little more complicated than the
remaining cases.
Let us fix $U\in AMp$ so that $\tilde\Omega=U\Omega$.
 Clearly, we can write
\begin{align}
H=\tilde{h}(b^{*},b)=B+\bar{K}_{i}b_{i}+K_{i}b_{i}^{*}+O_{ij}b^{*}_{j}b^{*}_{i}+\bar{O}_{ij}b_{i}b_{j}
+ D_{ij}b^{*}_{i}b_{j} + \nonumber \\ 
+ \text{terms of higher order in
  $b$'s}. 
\label{transH}
\end{align}  
We know that in a neighborhood of $\tilde\Omega$ arbitrary pure
 Gaussian states are
parametrized by a symmetric matrix $\theta$ and a vector $y$:
\[\theta\mapsto U \e^{i\phi(y)}\e^{iX_\theta}\Omega,\] 
where $X_\theta:=\theta_{ij}a_i^*a_j^*+\bar\theta_{ij}a_ja_i$ and $\phi(y)= y_{i}a_i^*+\bar y_{i} a_i$.
We get
\begin{eqnarray}
(U\e^{\i\phi(y)}\e^{\i X_\theta}\Omega|HU \e^{\i\phi(y)}e^{\i X_\theta}\Omega)&=&(\e^{\i\phi(y)}\e^{\i X_\theta}\Omega|U^{*}\tilde h(b^*,b)U\e^{\i\phi(y)}\e^{\i X_\theta}\Omega)\nonumber \\
&=&(\Omega|\e^{-\i X_\theta}\e^{-\i\phi(y)}\tilde h(a^*,a)\e^{\i\phi(y)}\e^{\i X_\theta}\Omega). \label{obroty} \end{eqnarray}
Now 
\begin{eqnarray*}
\e^{-\i X_\theta}\e^{-\i\phi(y)}\tilde h(a^*,a)\e^{\i\phi(y)}\e^{\i
  X_\theta}
&=&B- \i(\bar{\theta}_{ij}O_{ij}- \theta_{ij}\bar{O}_{ij})
-\i(\bar{y}_{i}K_{i}-y_{i}\bar{K}_{i})\\&&+\hbox{terms containing $a_i$ or $a_i^*$
  }+O(\|\theta\|^{2},\|y\|^{2}). \end{eqnarray*}
Therefore, (\ref{obroty}) equals
\begin{eqnarray}
B- \i(\bar{\theta}_{ij}O_{ij}- \theta_{ij}\bar{O}_{ij})
-\i(\bar{y}_{i}K_{i}-y_{i}\bar{K}_{i})+O(\|\theta\|^{2},\|y\|^{2}).
\label{obrot} 
\end{eqnarray}
Since vectors $y$ and matrices $\theta$ are independent variables,
 (\ref{obrot}) is stationary at $\tilde\Omega$ if and only if $[O_{ij}]$ is a
zero matrix and $[K_{i}]$ is a zero vector. This ends the proof of part (2).

To prove  (3) and (4) we note that
 for $U\in Pin$, the neighborhood of $\tilde{\Omega}=U\Omega$  in the set of
 fermionic pure Gaussian states is
parametrized by antisymmetric matrices $\theta$:
\[\theta\mapsto U \e^{iX_\theta}\Omega,\] 
where again
$X_\theta:=\theta_{ij}a_i^*a_j^*+\bar\theta_{ij}a_ja_i$. Therefore, it
suffices to repeat the above proof with $y_i=K_i=0$, $i=1,\dots,n$.

The proof of (1) is similar. \end{proof}

\begin{proposition} In addition to the assumptions of Theorem \ref{main} (2),
  suppose that
 $\tilde\Omega$ corresponds to a minimum. Then
the matrix $[D_{ij}] $ is positive.
\eep 

\begin{proof}
Using that
  $O$ and $K$ are zero, we obtain
\begin{eqnarray*}
\e^{-\i\phi(y)}\tilde h(a^*,a)\e^{\i\phi(y)}
&=&B
+\bar{y}_{i}D_{ij}y_j\\&&+\hbox{terms containing $a_i$ or $a_i^*$
  }+O(\|y\|^{3}). \end{eqnarray*}
Therefore, (\ref{obroty}) equals
\begin{eqnarray}
B+\bar{y}_{i}D_{ij}y_j+O(\|y\|^{3}).
\label{obrot1} 
\end{eqnarray}
Hence the matrix
$[D_{ij}]$ is positive.
\end{proof}

 Note that in cases (1), (3) and (4) the matrix $[D_{ij}]$ does not have
  to be positive.

\end{document}